\documentclass[preprint, 11pt]{elsarticle}

\usepackage{amssymb}
\usepackage{amsthm}
\usepackage{amsmath}
\usepackage[left=1.2in,top=1.2in,right=1.2in,bottom=1.2in,nohead]{geometry}

\newtheorem{theorem}{Theorem}
\newtheorem{lemma}{Lemma}

\newtheorem{proposition}{Proposition}

\usepackage{color}

\graphicspath{{figures/}}

\newcommand{\Setx}[1]{\left\{#1\right\}}

\newcommand{\dist}[2]{d(#1,#2)}
\newcommand{\norm}[1]{\|#1\|}
\newcommand{\RR}{\mathbb{R}}
\newcommand{\CD}[2][]{\ensuremath{\textup{\textsf{C-DISC}}_{#1}(#2)}}
\newcommand{\CGG}[1]{\ensuremath{\textup{\textsf{GG}}(#1)}}
\newcommand{\kGG}[1]{\ensuremath{\textup{\textsf{$k$-GG}}(#1)}}

\newcommand{\IGG}[1]{\ensuremath{\textup{\textsf{$1$-GG}}(#1)}}
\newcommand{\XGG}[1]{\ensuremath{\textup{\textsf{$10$-GG}}(#1)}}
\newcommand{\VGG}[1]{\ensuremath{\textup{\textsf{$5$-GG}}(#1)}}
\newcommand{\VIIGG}[1]{\ensuremath{\textup{\textsf{$7$-GG}}(#1)}}
\newcommand{\VIIIGG}[1]{\ensuremath{\textup{\textsf{$8$-GG}}(#1)}}

\newcommand{\XIXGG}[1]{\ensuremath{\textup{\textsf{$19$-GG}}(#1)}}
\newcommand{\XVGG}[1]{\ensuremath{\textup{\textsf{$15$-GG}}(#1)}}
\newcommand{\kDG}[1]{\ensuremath{\textup{\textsf{$k$-DG}}(#1)}}
\newcommand{\kIDG}[1]{\ensuremath{\textup{\textsf{$(k+1)$-DG}}(#1)}}
\newcommand{\DG}[1]{\ensuremath{\textup{\textsf{DG}}(#1)}}
\newcommand{\ODG}[1]{\ensuremath{\textup{\textsf{$0$-DG}}(#1)}}
\newcommand{\IDG}[1]{\ensuremath{\textup{\textsf{$1$-DG}}(#1)}}

\newcommand{\XIXDG}[1]{\ensuremath{\textup{\textsf{$19$-DG}}(#1)}}
\newcommand{\XVDG}[1]{\ensuremath{\textup{\textsf{$15$-DG}}(#1)}}
\newcommand{\RNG}[1]{\ensuremath{\textup{\textsf{RNG}}(#1)}}
\newcommand{\kRNG}[1]{\ensuremath{\textup{\textsf{$k$-RNG}}(#1)}}
\newcommand{\kImRNG}[1]{\ensuremath{\textup{\textsf{$(k-1)$-RNG}}(#1)}}

\newcommand{\XIXRNG}[1]{\ensuremath{\textup{\textsf{$19$-RNG}}(#1)}}

\journal{arXiv}

\begin{document}

\begin{frontmatter}



\title{10-Gabriel graphs are Hamiltonian}


\author[ZCU1]{Tom\'a\v{s} Kaiser\fnref{fn1}}
\ead{kaisert@kma.zcu.cz}
\author[ZCU2]{Maria Saumell\corref{cor}\fnref{fn2}}
\ead{saumell@kma.zcu.cz}
\author[ZCU2]{Nico Van Cleemput\fnref{fn2}}
\ead{cleemput@kma.zcu.cz}

\cortext[cor]{Corresponding author.}
\fntext[fn1]{Supported by project 14-19503S of the
Czech Science Foundation.}
\fntext[fn2]{Supported by the project NEXLIZ – CZ.1.07/2.3.00/30.0038, which is co-financed by the European Social Fund and the state budget of the Czech Republic.}

\address[ZCU1]{Department of Mathematics, Institute for Theoretical Computer Science (CE-ITI), and European Centre of Excellence NTIS (New Technologies for the Information Society), University of West Bohemia, Pilsen, Czech Republic}

\address[ZCU2]{Department of Mathematics and European Centre of Excellence NTIS (New Technologies for the Information Society), University of West Bohemia, Pilsen, Czech Republic}

\begin{abstract}
Given a set $S$ of points in the plane, the \emph{$k$-Gabriel graph} of $S$ is the geometric graph with vertex set $S$, where $p_i,p_j\in S$ are connected by an edge if and only if the closed disk having segment $\overline{p_ip_j}$ as diameter contains at most $k$ points of  $S \setminus \{p_i,p_j\}$. We consider the following question: What is the minimum value of $k$ such that the $k$-Gabriel graph of every point set $S$ contains a Hamiltonian cycle? For this value, we give an upper bound of 10 and a lower bound of 2. The best previously known values were 15 and 1, respectively.
\end{abstract}


\end{frontmatter}


\section{Introduction}
\label{sec:intro}

Let $S$ be a set of $n$ distinct points in the plane. Loosely speaking, a \emph{proximity graph} on $S$ is a graph that attempts to capture the relations of proximity among the points in $S$. Usually, one defines a reasonable criteria for two points to be considered close to each other, and then the pairs of points that satisfy the criteria are connected in the graph. The study of proximity graphs has been a popular topic in computational geometry, since these graphs not only satisfy interesting theoretical properties, but also have applications in several fields, such as shape analysis, geographic information systems, data mining, computer graphics, or graph drawing (see, for example,~\cite{bookVDDT,Toussaint-surv}).
  
The Delaunay graph and its relatives constitute a prominent family of proximity graphs. In the \emph{Delaunay graph} of $S$, denoted by $\DG{S}$, $p_i,p_j\in S$ are connected by an edge if and only if there exists a closed disk with $p_i,p_j$ on its boundary that does not contain any point of $S \setminus \{p_i,p_j\}$ (see~\cite{Delaunay}). It is well-known that if $S$ does not contain three collinear or four cocircular points, then $\DG{S}$ is a triangulation of $S$.

Two related proximity graphs are the relative neighborhood graph and the Gabriel graph. In the \emph{relative neighborhood graph} of $S$, denoted by $\RNG{S}$, $p_i,p_j\in S$ are connected by an edge if and only if there does not exist any $p_{\ell}\in S$ such that $\dist {p_i} {p_{\ell}} < \dist {p_i} {p_j}$ and $\dist {p_j} {p_{\ell}} < \dist {p_i} {p_j}$, where $\dist p q$
denotes the Euclidean distance between $p$ and $q$ (see~\cite{Touss-RNG}). 

Given two points $p_i,p_j\in S$, we denote the
closed disk having segment $\overline{p_ip_j}$ as diameter by $
\CD{p_i,p_j}$. The \emph{Gabriel graph} of $S$ is
the graph in which $p_i,p_j\in S$ are connected by an edge if and only if $\CD{p_i,p_j}\cap
S=\{p_i,p_j\}$ (see~\cite{GabSok:gabgraph}). We denote the Gabriel graph of $S$ by $\CGG{S}$. Notice that $\RNG{S}\subseteq \CGG{S} \subseteq \DG{S}$ holds for any point set $S$.

All of the above graphs are plane, that is, if edges are drawn as line segments, then the resulting drawing contains no crossings. In the last decades, a number of works have been devoted to investigate whether they fulfill other desirable graph-theoretic, geometric, or computational properties. For example, it has been studied whether the vertices of these graphs have bounded maximum or expected degree~\cite{MatSok,Devr-maxdegGG,table-paper}, whether these graphs are constant spanners~\cite{SRGG,Dobkin-DG}, or whether they are compatible with simple online routing algorithms~\cite{CompRout-Urrutia}.

A problem that attracted much attention is the Hamiltonicity of Delaunay graphs: Does $\DG{S}$ contain a Hamiltonian cycle for every point set $S$? Dillencourt~\cite{nonHamDT} answered this question negatively by providing an example of a set of points whose Delaunay graph is a non-Hamiltonian triangulation. This naturally raises the question whether there exist variants of the Delaunay graph that do always contain a Hamiltonian cycle.

This problem has been studied for the Delaunay graph in the $L_{\infty}$ metric. This graph contains an edge between $p_i,p_j\in S$ if and only if there exists an axis-aligned square containing $p_i,p_j$ and no other point in $S$. Even though Delaunay graphs in the $L_{\infty}$ metric need not contain a Hamiltonian cycle, they satisfy the slightly weaker property of containing a Hamiltonian path, as shown by {\'A}brego \textit{et al.}~\cite{MatchPtSq}.

Another natural variant of Delaunay graphs which has received some interest is that of \emph{$k$-Delaunay graphs}, $\kDG{S}$ for short~\cite{kDG-Abe}. In this case, the definition is relaxed in the following way: $p_i,p_j\in S$ are connected by an edge if and only if there exists a closed disk with $p_i,p_j$ on its boundary that contains at most $k$ points of $S \setminus \{p_i,p_j\}$. Analogous generalizations lead to \emph{$k$-Gabriel graphs} and \emph{$k$-relative neighborhood graphs}. The $k$-Gabriel graph of $S$, denoted by $\kGG{S}$, is the graph
in which $p_i,p_j$ are connected by an edge if and only if $|\CD{p_i,p_j}\cap S|\leq k+2$ (see~\cite{kGG-Su}). The $k$-relative neighborhood graph of $S$, denoted by $\kRNG{S}$, is the graph
in which $p_i,p_j$ are connected by an edge if and only if there exist at most $k$ points $p_{\ell}\in S$ such that $\dist {p_i} {p_{\ell}} < \dist {p_i} {p_j}$ and $\dist {p_j} {p_{\ell}} < \dist {p_i} {p_j}$ (see~\cite{CTL:kRNG}).

Notice that $\ODG{S}=\DG{S}$ and, for any $k\geq 0$, $\kDG{S} \subseteq \kIDG{S}$. Since $\kDG{S}$ is the complete graph for $k\geq n/2$~\cite{kDG-Abe} and the complete graph is Hamiltonian, the following question arises: What is the minimum value of $k$ such that 
$\kDG{S}$ is Hamiltonian for every $S$? Abellanas \textit{et al.}~\cite{kDG-Abe} conjectured that this value is 1, that is, $\IDG{S}$ is already Hamiltonian. The same question can be formulated for $\kGG{S}$ and $\kRNG{S}$.

The first upper bound for such minimum value of $k$ was given by Chang \textit{et al.}~\cite{20RNGHam}, who proved that $\XIXRNG{S}$ is always Hamiltonian\footnote{Chang \textit{et al.}~\cite{20RNGHam} define $\kRNG{S}$ in a slightly different way, so $\kRNG{S}$ in their paper is equivalent to $\kImRNG{S}$ in our paper.}. Since, for any $k\geq 0$, $\kRNG{S}\subseteq \kGG{S} \subseteq \kDG{S}$, the result implies that $\XIXGG{S}$ and $\XIXDG{S}$ are also Hamiltonian. Later, Abellanas \textit{et al.}~\cite{kDG-Abe} improved the bound for the latter graphs by showing that $\XVGG{S}$ (and thus $\XVDG{S}$) is already Hamiltonian\footnote{There also exists an unpublished upper bound of 13~\cite{perouz}.}. In this short paper we improve their bound as follows:

\begin{theorem}\label{main-theorem}
For any set of points $S$, the graph $\XGG{S}$ is Hamiltonian.
\end{theorem}

We note that related properties of $k$-Gabriel graphs have been recently considered by Biniaz \textit{et al.}~\cite{BiniazMS14a}. In particular, the authors show that $\XGG{S}$ always contains a Euclidean bottleneck perfect matching, that is, a perfect matching that minimizes the length of the longest edge. Our proof of Theorem 1 actually shows that 10-GG(S) always contains a Euclidean bottleneck Hamiltonian cycle, which is a Hamiltonian cycle minimizing the length of the longest edge. Even though the two results are closely related, there is no direct implication between them.

We prove Theorem~\ref{main-theorem} in Section~\ref{sec:main-result}. Our proof uses the same general strategy as the ones in~\cite{kDG-Abe,20RNGHam}: We select a particular Hamiltonian cycle of the complete graph on $S$, and we find a value of $k$ such that $\kDG{S}$ contains this Hamiltonian cycle. In Section~\ref{sec:remarks}, we show that the best result that can possibly be proved with this particular approach is the Hamiltonicity of 6-Gabriel graphs (we also indicate the existence of an unpublished example~\cite{example-Biniaz} showing that the method cannot go beyond \ensuremath{\textup{\textsf{$8$-GG}}}). We further point out that it might be possible to decrease the value 10 by using a quadratic solver. Finally, we provide an example showing that 1-Gabriel graphs are not always Hamiltonian.

\section{Proof of Theorem~\ref{main-theorem}}
\label{sec:main-result}

The first steps of our proof go along the same lines as the arguments in~\cite{kDG-Abe} showing that 15-Gabriel graphs are Hamiltonian. The same general strategy was first used in~\cite{20RNGHam}. We provide the details for completeness.

We denote by $\mathcal{H}$ the set of all Hamiltonian cycles of the complete graph on $S$. Given a cycle $h\in \mathcal{H}$, we define the \emph{distance sequence} of $h$, denoted $ds(h)$, as the sequence containing the lengths of the edges of $h$ sorted in decreasing order (the \emph{length} of an edge is the length of the straight-line segment connecting its endpoints). Then, we define a strict order on the elements of $\mathcal{H}$ as follows: for $h_1,h_2\in \mathcal{H}$, we say that $h_1\succ h_2$ if and only if $ds(h_1)>ds(h_2)$ in the lexicographical order. 

Let $m$ be a minimal element of $\mathcal{H}$ with respect to the order that we have just defined. 
In the remainder of this section we show that all edges of $m$ belong to $\XGG{S}$, which in particular implies that $\XGG{S}$ is Hamiltonian.

Let $e=xy$ be any edge of $m$. We are going to show that $e$ is in $\XGG{S}$. Without loss of generality, we suppose that $x=(-1,0)$ and $y=(1,0)$. For any point $p$ in $\RR^2$, we write $\norm p$ for
the distance of $p$ from the origin $o=(0,0)$.

Let $U=\left\{u_1,u_2,\ldots,u_{\kappa} \right\}$ be the set of points in $S$ different from $x,y$ that are contained in $\CD{x,y}$. We want to prove that $\kappa \leq 10$. Suppose that, if we traverse the entire cycle $m$ starting from the ``directed" edge $\overrightarrow{xy}$ and finishing at $x$, we encounter the vertices of $U$ in the order $u_1,u_2,\ldots,u_{\kappa}$. For each point $u_i$, we denote by $s_i$ the point in $S$ preceding $u_i$ in this traversal of $m$ (see Figure~\ref{fig:cycle-m}). Note that possibly $s_1=y$.

\begin{figure}[htb]
    \centering
    \includegraphics[scale=0.7]{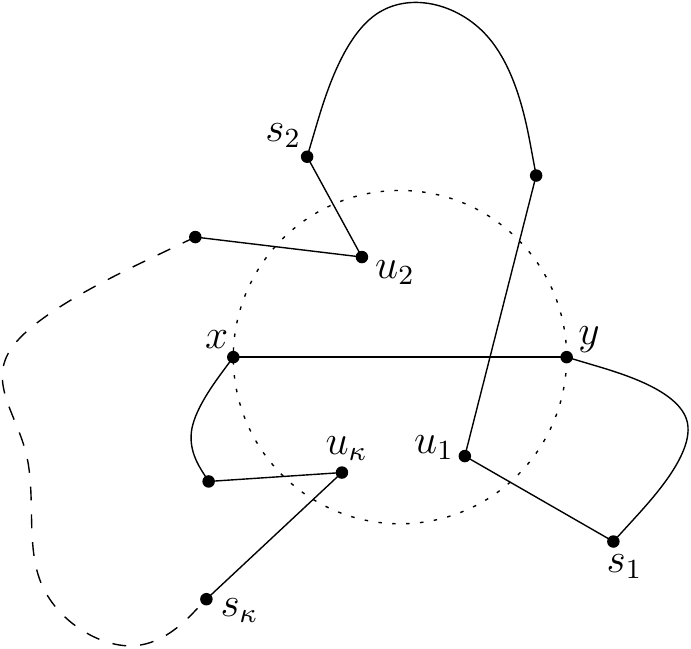}
    \caption{Cycle $m$ and the points $u_i$ and $s_i$.}
    \label{fig:cycle-m}
\end{figure}

We first prove that the following inequality holds, for $1\leq i\leq \kappa$:
\begin{equation}
  \dist{s_i}x \geq  \max\left\{\dist{s_i}{u_i},2\right\}.\label{eq:sx}
\end{equation}

(We stress that the maximum on the right hand side is \emph{not} taken over varying values of~$i$.)

If $s_1=y$, then $\dist{s_1}x=2$ and $\dist{s_1}{u_1}< 2$, so the inequality is satisfied. Otherwise, consider the Hamiltonian cycle $m'$ obtained by removing edges $s_iu_i$ and $xy$ from $m$, and adding edges $s_ix$ and $u_iy$. Note that, since $u_i$ lies in $\CD{x,y}$, we have that $\dist{u_i}{y}<\dist{x}{y}=2$. If $\dist{s_i}x < \max\left\{\dist{s_i}{u_i},2\right\}$, then it implies that $\max \left\{\dist{s_i}{x},\dist{u_i}{y}\right\}<\max\left\{\dist{s_i}{u_i},\dist{x}{y} \right\}$. Thus we would obtain that $m\succ m'$, contradicting the minimality of $m$. Hence we conclude that $\dist{s_i}x \geq \max\left\{\dist{s_i}{u_i},2\right\}$.

We observe that inequality~(\ref{eq:sx}) implies that, except for the case when $s_1=y$, the points $s_i$ are outside $\CD{x,y}$, as depicted in Figure~\ref{fig:cycle-m}. In particular, it is not possible that $u_i=s_{i+1}$ for any $i$.

Next, let $1\leq i<j\leq \kappa$. We show that the following inequality holds:
\begin{equation}
  \dist{s_i}{s_j} \geq \max\left\{\dist{s_i}{u_i},\dist{s_j}{u_j},2\right\}.\label{eq:ss}
\end{equation}

Suppose, for the sake of contradiction, that $\dist{s_i}{s_j} < \max\left\{\dist{s_i}{u_i},\dist{s_j}{u_j},2 \right\}$. We consider the Hamiltonian cycle $m''$ obtained by removing edges $s_iu_i$, $s_ju_j$ and $xy$ from $m$, and adding edges $s_is_j$, $u_ix$ and $u_jy$. As in the previous case, we have that $\dist{u_i}{x}<2$ and $\dist{u_j}{y}<2$. We obtain that $\max \left\{\dist{s_i}{s_j}, \dist{u_i}{x},\dist{u_j}{y}\right\}<\max\left\{\dist{s_i}{u_i},\dist{s_j}{u_j},\dist{x}{y} \right\}$. Thus $m \succ m''$, which contradicts the minimality of $m$.

Abellanas \textit{et al.}~\cite{kDG-Abe} use inequalities~(\ref{eq:sx}) and~(\ref{eq:ss}), together with some other geometric observations, to derive the bound $\kappa \leq 15$. Essentially, their argument consists of dividing the plane into several regions, and proving that each region contains at most one point of type $s_i$. We now present a packing-based argument that allows to reduce the upper bound to 10.

For a point $x$ in $\RR^2$ and a positive $r$, let $D(x,r)$ be the closed
disk with center $x$ and radius $r$. Additionally, we denote the boundary of this disk by $\partial D(x,r)$; in other words, $\partial D(x,r)$ is the circle of radius $r$ centered at
$x$.

For $i=1,\dots,\kappa$,
we define $s'_i$ as the intersection point between $\partial D(o,3)$ and the ray with
  origin at $o$ and passing through $s_i$ (i.e., $s'_i$ is the \emph{projection} of $s_i$ to $\partial D(o,3)$). If $\norm{s_i} >
  3$,
we define $D_i$ as the unit disk (i.e., the disk of radius 1) centered at $s'_i$; otherwise, $D_i$ is the unit disk centered at $s_i$.
Finally, we denote the unit disk centered at $x$ by $D_0$.

\begin{lemma}\label{l:packing}
  All the disks $D_i$, where $0\leq i \leq \kappa$, are pairwise
  internally disjoint.
\end{lemma}
\begin{proof}
  We consider two disks $D_i$, $D_j$ $(0 \leq i, j \leq \kappa)$ and
  distinguish the possible cases with respect to the types of $D_i$ and $D_j$.

  Suppose first that either $i$ or $j$, for example $i$, equals 0. Thus the center of $D_i$ is $x$. If $\norm{s_j}\leq 3$,
  then $D_j$ is centered at $s_j$ and therefore is internally disjoint
  from $D_i$ by~(\ref{eq:sx}). On the other hand, if $\norm{s_j}>3$,
  then the center $s'_j$ of $D_j$ is on $\partial D(o,3)$ and $\dist{s'_j}{x}\geq
  2$, which makes $D_i$ and $D_j$ internally disjoint.

  Suppose next that $i>0$, $j>0$, and at least one of the two inequalities $\norm{s_i}\leq 3$ and $\norm{s_j}\leq 3$, for example the first one, holds. If $\norm{s_j}\leq 3$, then
  $D_i$ and $D_j$ are centered at $s_i$ and $s_j$, respectively, so
  they are internally disjoint by~(\ref{eq:ss}). Let us consider the
  case where $\norm{s_j} > 3$. By~(\ref{eq:ss}), $s_i$ is not contained
  in the interior of the disk $D(s_j,\dist{s_j}{u_j})$. Since $u_j$ is
  contained in $D(o,1)$, $s_i$ is not contained in the interior of
  $D(s_j,\norm{s_j}-1)$. Note that the latter disk contains the disk
  $D(s'_j,2)$.
	Consequently, $\dist{s_i}{s'_j} \geq 2$, and $D_i$ and $D_j$ are
  internally disjoint.

  Finally, suppose that $i,j>0$, $\norm{s_i} > 3$ and $\norm{s_j} > 3$. Without loss of generality,
  we may assume that $\norm{s_i} \geq \norm{s_j}$. We prove that $D_i$ and $D_j$ are
  internally disjoint by contradiction. Since in this case $D_i$ and $D_j$ are respectively centered at $s'_i$ and $s'_j$, if the disks are not disjoint we get that $\dist{s'_i}{s'_j} < 2$. Since $s'_i$
  and $s'_j$ lie on $\partial D(o,3)$, for the angle
  $\alpha=s_ios_j$ we have that $\sin(\alpha/2)<\tfrac13$. Thus we easily find that $\cos\alpha > \tfrac79$. By the law of cosines,
  \begin{equation*}
    \dist{s_i}{s_j}^2 < \norm{s_i}^2+\norm{s_j}^2 -
    \frac{14}9\norm{s_i}\norm{s_j}.
  \end{equation*}
  On the other hand, by~(\ref{eq:ss}) we know that $\dist{s_i}{s_j}
  \geq \dist{s_i}{u_i}$. By the triangle inequality, $\norm{s_i} =\dist{s_i}{o}\leq \dist{s_i}{u_i}+\dist{u_i}{o}$. Since $\dist{u_i}{o}\leq 1$, we obtain that $\dist{s_i}{u_i} \geq \norm{s_i}-1$. Combining $\dist{s_i}{s_j} \geq \norm{s_i}-1$ with the previous inequality, it gives
  \begin{equation*}
    \norm{s_i}\left(\frac{14}9\norm{s_j}-2\right) < \norm{s_j}^2 - 1.
  \end{equation*}
  Using the assumption that $\norm{s_i}\geq\norm{s_j}$, we find
  \begin{equation*}
    \label{eq:z}
    \frac59\norm{s_j}^2-2\norm{s_j}+1 < 0.
  \end{equation*}
  To satisfy this inequality, $\norm{s_j}$ has to be contained in the interval
  $(\frac35,3)$, contradicting the assumption that
  $\norm{s_j}>3$. This completes the proof.
\end{proof}

The center of each of the unit disks $D_i$ ($0 \leq i \leq \kappa$) lies
within distance 3 of the origin, so by Lemma~\ref{l:packing},
$\Setx{D_0,\dots,D_{\kappa}}$ is a unit disk packing inside the circle
$\partial D(o,4)$. By a result of Fodor~\cite{fodor}, the smallest radius $R$ of a
circle admitting a packing of twelve unit disks satisfies $R >
4.029$. Since the radius of $\partial D(o,4)$ is $4$, we obtain that $\Setx{D_0,\dots,D_{\kappa}}$ is a unit disk packing of at most eleven disks, i.e., $\kappa +1 \leq 11$. Therefore, $\kappa \leq 10$, which
finishes the proof of Theorem~\ref{main-theorem}.

\section{Concluding remarks}
\label{sec:remarks}

In this section, we discuss a possible way to further improve upon Theorem~\ref{main-theorem}, as well as constructions showing lower bounds (both for the specific method and in general).

We start by making further observations about the minimal cycle $m$. For each point $u_i$, we denote by $t_i$ the point in $S$ succeeding $u_i$ in the traversal of $m$ starting from the ``directed" edge $\overrightarrow{xy}$ and finishing at $x$. Notice that possibly $t_{\kappa}=x$, or $t_i=s_{i+1}$ for some $1\leq i\leq \kappa -1$. As shown in~\cite{kDG-Abe}, by traversing $m$ in the reverse order and arguing as in~(\ref{eq:sx}), we obtain that, for $1\leq i\leq \kappa$,
\begin{equation}
  \dist{t_i}y \geq \max\left\{\dist{t_i}{u_i},2\right\}.\label{eq:ty}
\end{equation}

Additionally, we have an inequality involving distances between points of the form $t_i$ that is analogous to~(\ref{eq:ss}) (see~\cite{kDG-Abe}). For $1\leq i<j\leq \kappa$, we have:
\begin{equation}
  \dist{t_i}{t_j} \geq \max\left\{\dist{t_i}{u_i},\dist{t_j}{u_j},2\right\}.\label{eq:tt}
\end{equation}

We can also derive some inequalities involving distances between points of the form $s_i$ and points of the form $t_i$. First, for $1\leq i\leq \kappa$, we show:
\begin{equation}
  \dist{s_i}{t_i} \geq \max\left\{\dist{s_i}{u_i},\dist{t_i}{u_i},2\right\}.\label{eq:st1}
\end{equation}

If the inequality was not satisfied, we would have that $m \succ m'''$, where $m'''$ is the Hamiltonian cycle obtained by removing edges $s_iu_i$, $t_iu_i$ and $xy$ from $m$, and adding edges $s_it_i$, $u_ix$ and $u_iy$. 

Next, for $1\leq i<j\leq \kappa$, we can easily prove:
\begin{equation}
  \dist{s_i}{t_j} \geq \max\left\{\dist{s_i}{u_i},\dist{t_j}{u_j},2\right\}.\label{eq:st2}
\end{equation}

In this case, the Hamiltonian cycle used to prove the inequality is the one obtained by removing edges $s_iu_i$, $t_ju_j$ and $xy$ from $m$, and adding edges $s_it_j$, $u_ix$ and $u_jy$. 

For every point $u_i$, we define $u^x_i$ and $u^y_i$, respectively, as the $x$- and $y$-coordinates of $u_i$. In the same way, we define variables for the points of the form $s_i$ and $t_i$. Then we set $\mathcal{V}=\left\{u^x_i, u^y_i,s^x_i, s^y_i,t^x_i, t^y_i\,|\, 1\leq i\leq \kappa \right\}$. Since the points $u_i$ lie in $\CD{x,y}$, we have 

\begin{equation}
  (u^x_i)^2+(u^y_i)^2 \leq 1.\label{eq:u}
\end{equation}

Inequalities~(\ref{eq:sx})-(\ref{eq:u}) can be expressed as quadratic inequalities with variables in $\mathcal{V}$. Therefore, it might be possible to improve Theorem~\ref{main-theorem} by answering the following question: What is the maximum value of $\kappa$ such that inequalities~(\ref{eq:sx})-(\ref{eq:u}) define a non-empty region of $\RR^{6\kappa}$? Unfortunately, some of the constraints in the program are not convex, and our attempts to answer this question by using a quadratic programming solver have so far been unsuccessful. 

On the other hand, Figure~\ref{fig:six-triples} shows an example of a Hamiltonian cycle with an edge not in $\VGG{S}$, and which is minimal in $\mathcal{H}$ (we prove this in the next paragraph). This proves that the system of inequalities~(\ref{eq:sx})-(\ref{eq:u}) is feasible for $\kappa=6$. We conclude that, with this particular approach (what is the smallest value of $k$ such that all edges of any minimal Hamiltonian cycle belong to $\kGG{S}$?), the best result that one can possibly prove is that 6-Gabriel graphs are Hamiltonian. (In fact, Biniaz \textit{et al.}~\cite{example-Biniaz} further improved this by constructing a point set $S$ whose unique minimal Hamiltonian cycle is not contained in $\VIIGG{S}$, implying that the best possible result is that $\VIIIGG{S}$ is Hamiltonian.)

\begin{figure}[htb]
    \centering
    \includegraphics[scale=0.7]{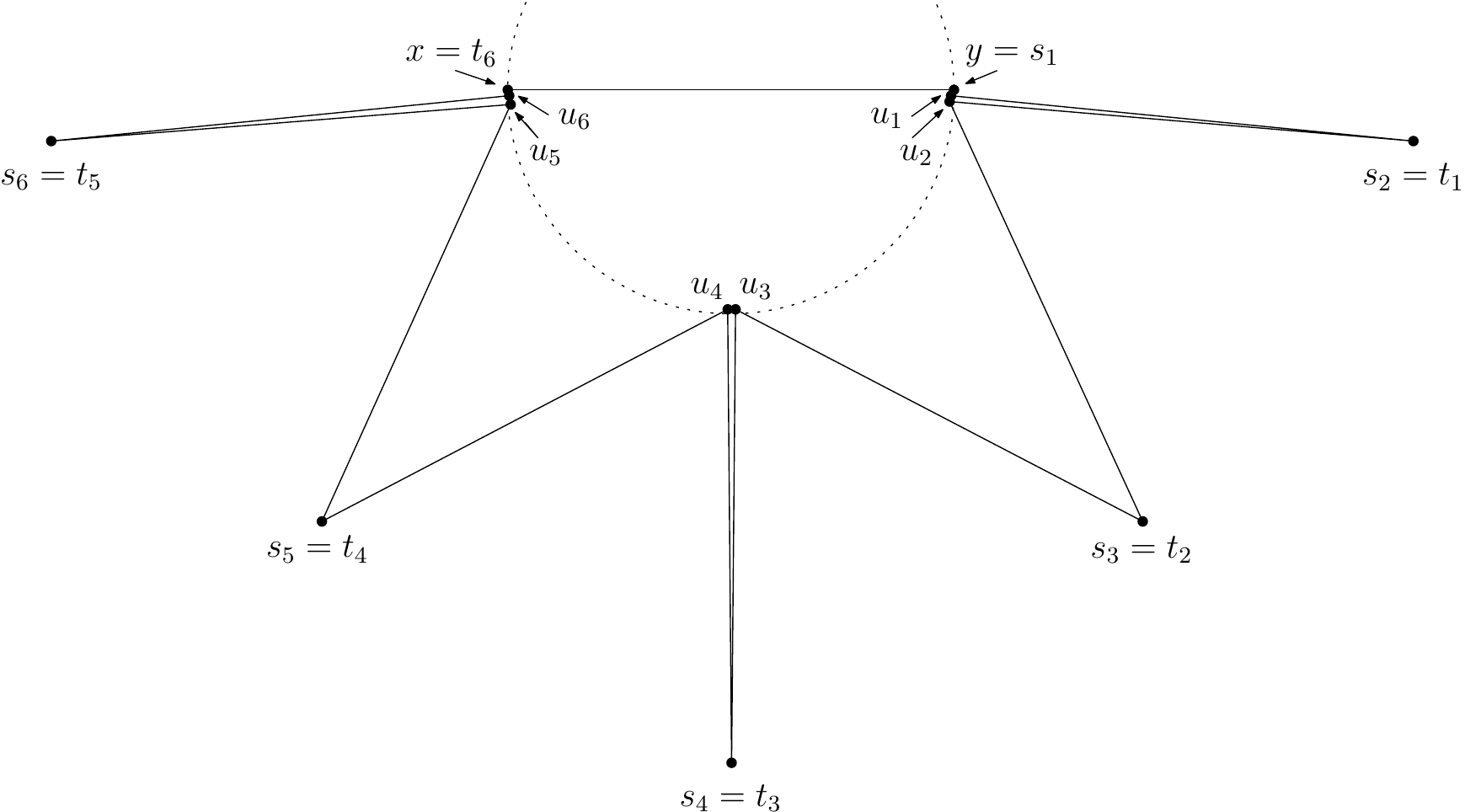}
    \caption{Minimal Hamiltonian cycle where one of the edges does not belong to $\VGG{S}$.}
    \label{fig:six-triples}
\end{figure}

In order to prove that the edges in Figure~\ref{fig:six-triples}  form a Hamiltonian cycle $h$ that is minimal, we point out that points have been arranged so that points $s_2,s_3,\ldots,s_6$ are connected to their two closest points in the point set. Now, $s_2u_1$ and $s_2u_2$ are the longest edges in the cycle, together with $s_6u_5$ and $s_6u_6$. Since $u_1$ and $u_2$ are the two closest points to $s_2$, any Hamiltonian cycle $h'$ where $s_2$ is not connected to $u_1$ or $u_2$ satisfies $ds(h')>ds(h)$. Thus, if there exists a cycle $h''$ such that $ds(h'')<ds(h)$, then $h''$ contains $s_2u_1$ and $s_2u_2$, and analogously $s_6u_5$ and $s_6u_6$. Similarly, we first find that edges $s_3u_2$, $s_3u_3$, $s_5u_4$ and $s_5u_5$ are also contained in $h''$, and then that $h''$ additionally contains $s_4u_3$ and $s_4u_4$. To conclude, it is easy to see that $h''$ contains $xy$, $xu_6$ and $yu_1$, obtaining the contradiction $h''=h$.

Finally, we give a lower bound for the minimum value of $k$ such that $k$-Gabriel graphs are always Hamiltonian. To the best of our knowledge, the only bound that was known is~1, which follows trivially from the fact that $0$-Delaunay graphs do not necessarily contain a Hamiltonian cycle~\cite{nonHamDT}. In the following proposition, we slightly improve this bound to 2:

\begin{proposition}
There exist point sets $S$ such that $\IGG{S}$ is not Hamiltonian. 
\end{proposition}

\begin{proof}
A very simple example of this fact is shown in Figure~\ref{fig:nonHam1GG}. We note that it is not difficult to produce examples involving larger point sets.
\end{proof}

\begin{figure}[htb]
    \centering
    \includegraphics[scale=0.65]{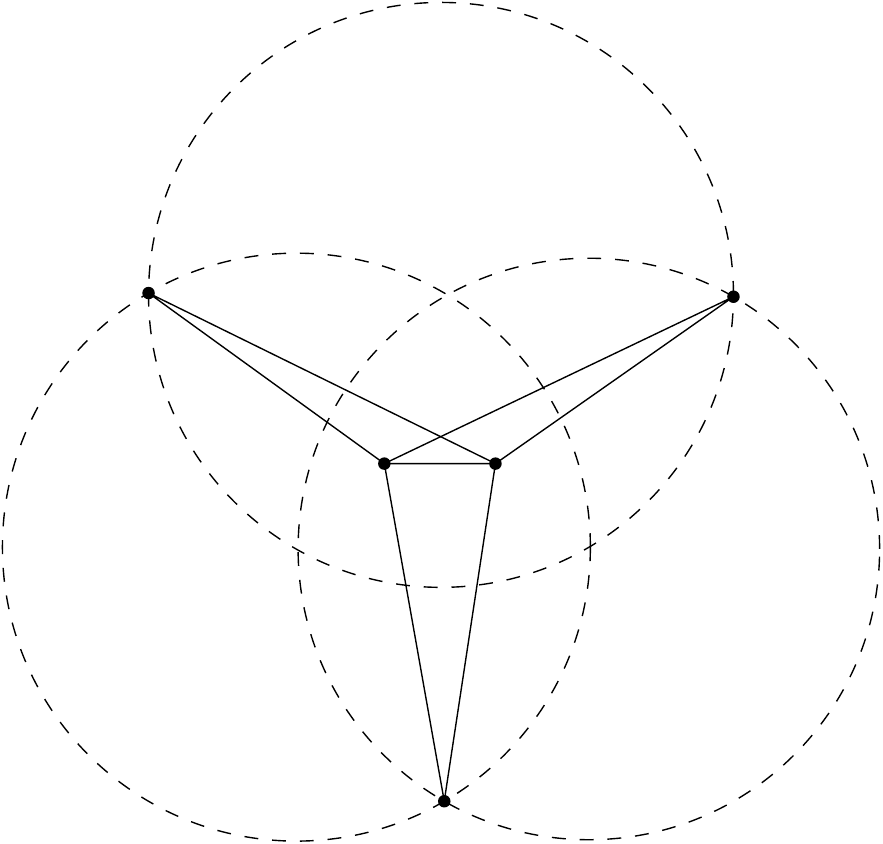}
    \caption{A point set $S$ and its 1-Gabriel graph, which is not Hamiltonian. The dashed circles show that the edges connecting the two points on the circles do not belong to $\IGG{S}$.}
    \label{fig:nonHam1GG}
\end{figure}

\paragraph{Acknowledgments} We would like to thank an anonymous reviewer for suggesting the example in Figure~\ref{fig:nonHam1GG}, which is smaller than our original construction.



\bibliographystyle{elsarticle-harv}
\bibliography{refs}









\end{document}